\documentclass[letterpaper,12pt]{llncs}

\usepackage{geometry}
\usepackage{graphicx}
\usepackage{subfigure}
\usepackage{epsfig}
\usepackage{amsmath}
\usepackage{amssymb}
\usepackage{algorithm}
\usepackage[noend]{algorithmic}

\usepackage{versions}
\excludeversion{LONG}
\excludeversion{REMOVED}

\renewcommand{\log}{\lg}

\def\idtt#1{\ensuremath{\mathtt{#1}}}

\def\accop{\idtt{access}}
\def\rankop{\idtt{rank}}
\def\selop{\idtt{select}}
\def\insop{\idtt{insert}}
\def\delop{\idtt{delete}}
\def\sumop{\idtt{sum}}
\def\searchop{\idtt{search}}
\def\updateop{\idtt{update}}
\def\predop{\idtt{pred}}
\def\succop{\idtt{succ}}

\def\etal{{\em et al.}}

\begin{document}

\title{Compressed Dynamic Range Majority and Minority Data Structures
\thanks{Partially supported by Fondecyt grant 1-171058, Chile; NSERC, Canada; basal funds FB0001, Conicyt, Chile; and the Millenium Institute for Foundational Research on Data, Chile. A preliminary partial version of this article appeared in {\em Proc. DCC 2017} \cite{GHN17}.}}
\author{Travis Gagie\inst{1}\fnmsep\inst{2} \and Meng He\inst{3} \and Gonzalo Navarro\inst{1}\fnmsep\inst{4}\fnmsep\inst{5}}
\institute{
CeBiB --- Center for Biotechnology and Bioengineering, Chile
\and
School of Computer Science and Telecommunications, Diego Portales University, Chile, \email{travis.gagie@gmail.com}  
\and
Faculty of Computer Science, Dalhousie University, Canada, \email{mhe@cs.dal.ca} 
\and
Millenium Institute for Foundational Research on Data, Chile
\and
Department of Computer Science, University of Chile, Chile, \email{gnavarro@dcc.uchile.cl}  
}

\maketitle

\begin{abstract}
  \normalsize
In the range $\alpha$-majority query problem, we are given a sequence $S[1..n]$ and a fixed threshold $\alpha \in (0, 1)$, and are asked to preprocess $S$ such that, given a query range $[i..j]$, we can efficiently report the symbols that occur more than $\alpha (j-i+1)$ times in $S[i..j]$, which are called the range $\alpha$-majorities.
  In this article we first describe a dynamic data structure that represents $S$ in
compressed space --- $nH_k+ o(n\lg \sigma)$ bits for any $k =
o(\log_{\sigma} n)$, where $\sigma$ is the alphabet size and $H_k \le H_0 \le
\lg\sigma $ is the
$k$-th order empirical entropy of $S$ --- and answers queries in $O \left(
\frac{\log n}{\alpha \log \log n} \right)$ time while supporting insertions and
deletions in $S$ in $O \left( \frac{\lg n}{\alpha} \right)$ amortized time.
  We then show how to modify our data structure to receive some $\beta \ge
\alpha$ at query time and report the range $\beta$-majorities in $O \left( \frac{\log n}{\beta \log \log n} \right)$ time, without increasing the asymptotic space or update-time bounds.
  The best previous dynamic solution has the
same query and update times as ours, but it occupies $O(n)$ words and cannot
take advantage of being given a larger threshold $\beta$ at query time. 

Not even static data structures have previously achieved compression in terms of higher-order entropy.
The smallest ones take $n H_0 + o(n) (H_0 + 1)$ bits and answer queries in $O (f (n) / \alpha)$ time or take $(1 + \epsilon) n H_0 + o (n)$ bits and answer queries in optimal $O (1 / \alpha)$ time, where $f (n)$ is any function in $\omega (1)$ and $\epsilon$ is any constant greater than 0.
By giving up updates, we can improve our query
time to $O((1 / \alpha) \log \log_w \sigma)$ on a RAM with word size $w =
\Omega(\log n)$ bits, which is only slightly suboptimal, without increasing our space bound. 
Finally, we design the first dynamic data structure for range
$\alpha$-minority --- i.e., find a non-$\alpha$-majority that occurs in a 
range --- and obtain space and time bounds similar to those for 
$\alpha$-majorities. A static version of this structure is also the first
$\alpha$-minority data structure achieving compression in terms of $H_k$.
\end{abstract}

\section{Introduction}
\label{sec:intro}

An $\alpha$-majority in a sequence $S[1..n]$ is a character that occurs more than $\alpha n$ times in $S$, where the threshold $\alpha \in (0, 1)$.  Misra and Gries~\cite{mg1982} proposed a two-pass algorithm for finding all $\alpha$-majorities that runs in $O (1 / \alpha)$ space and can be made to run in linear time~\cite{dlm2002}.  In contrast, any algorithm that makes only a constant number of passes over $S$ needs nearly linear space even to estimate the frequency of the mode well~\cite{bjr07}, where the mode of $S$ is defined as its most frequent element.  Thus, finding $\alpha$-majorities is often considered a practical way to find frequent characters in large files and is important in data mining~\cite{fsmmu1998,dlm2002,ksp03}, for example.

For the {\em range $\alpha$-majority query} problem, we are asked to preprocess $S$ such that, given a query range $[i..j]$, we can efficiently report the $\alpha$-majorities of $S[i..j]$, i.e., the symbols that occur more than $\alpha (j-i+1)$ times in $S[i..j]$.
Not surprisingly, this problem seems easier than the range mode query problem~\cite{gglt10,cdlmw14}, in which the query asks for the most frequent element in the query range.
Karpinski and Nekrich~\cite{kn2008} first considered the range $\alpha$-majority query problem and proposed a solution that uses $O(n/\alpha)$ words to support queries in $O((\lg\lg n)^2/\alpha)$ time.
  Durocher~\etal~\cite{dhmn2013} presented the first solution that achieves optimal $O(1/\alpha)$ query time, and their structure also occupies $O(n/\alpha)$ words.
  Subsequent researchers have worked to make the space usage independent of $\alpha$~\cite{ghmn2011,cdsw2015,bgmnn2016} and even to achieve compression~\cite{ghmn2011,bgmnn2016}. 
  Among all these works, the most recent one is that of Belazzougui~\etal~\cite{BGN13,bgmnn2016}, who showed how to represent $S$ using $(1+\epsilon)nH_0 + o(n)$ bits for any constant $\epsilon > 0$ to answer range $\alpha$-majority queries in $O(1/\alpha)$ time, where $H_0$ is the $0$-th order empirical entropy of $S$.
  When more compression is desired, they also showed how to represent $S$ in $nH_0 + o(n)(H_0+1)$ bits to support range $\alpha$-majority in $O(f(n)/\alpha)$ time, for any $f(n) = \omega(1)$.
  Their solutions work for {\em variable} $\alpha$, that is, 
$\alpha$ is not known at construction time; the value of $\alpha$ is given
together with the range $[i,j]$ in each query.
  We refer readers to their most recent paper~\cite{bgmnn2016} for a more thorough survey.

  In the dynamic setting, we wish to maintain support for range
$\alpha$-majority queries under the following update operations on $S$: i)
$\insop(c, i)$, which inserts symbol $c$ between $A[i-1]$ and $A[i]$, shifting
the symbols in positions $i$ through $n$ to positions $i+1$ through $n+1$,
respectively; ii) $\delop(c,i)$, which deletes $A[i]$, shifting the symbols in
positions $i$ through $n$ to positions $i-1$ through $n-1$, respectively.
Elmasry~\etal~\cite{ehmn2016} considered this setting, and designed an
$O(n)$-word structure that can answer range $\alpha$-majority queries in
$O(\frac{\lg n}{\alpha \lg\lg n})$ time, supporting insertions and deletions in $O(\frac{\lg n}{\alpha})$ amortized time. 
  Before their work, Karpinski and Nekrich~\cite{kn2008} also considered the dynamic case, though they defined the dataset as a set of colored points in 1D. With a proper reduction \cite{ehmn2016}, the solutions by  Karpinski and Nekrich can also be used to encode dynamic sequences, although the results are inferior to those of Elmasry~\etal~\cite{ehmn2016}.
  More precisely, their data structures, when combined with the reduction \cite{ehmn2016}, can represent $S$ in $O(n/\alpha)$ words of space, answer queries in time $O(\frac{\lg^2 n}{\alpha})$, and support insertions and deletions in $O(\frac{\lg^2 n}{\alpha})$ amortized time.
  Alternatively, they can increase the space cost to $O(\frac{n\lg n}{\alpha})$, while decreasing the query and update times to  $O(\frac{\lg n}{\alpha})$  worst-case and amortized time, respectively.
  All the previous work for the dynamic case requires $\alpha$ to be a fixed value given at construction time. 

  A closely related problem is the {\em range $\alpha$-minority query}
problem, in which we preprocess a sequence $S$ such that, given a query range
$[i..j]$, we can efficiently report one {\em $\alpha$-minority} of $S[i..j]$,
i.e., a symbol that occurs at least once but not more than $\alpha (j-i+1)$
times in $S[i..j]$, if such a symbol exists, and otherwise return that there is no $\alpha$-minority in the range. 
  Chan~\etal~\cite{cdsw2015} studied this problem and designed an $O(n)$-word
data structure that answers range $\alpha$-minority queries in $O(1/\alpha)$ time.
  Belazzougui~\etal~\cite{BGN13,bgmnn2016} further designed succinct data structures for range $\alpha$-minority.
  They again presented two tradeoffs: they either represent $S$ using
$(1+\epsilon)nH_0 + o(n)$ bits for any constant $\epsilon > 0$ to answer range
$\alpha$-minority queries in $O(1/\alpha)$ time, or use $nH_0 + o(n)(H_0+1)$
bits and support range $\alpha$-minority queries in $O(f(n)/\alpha)$ time, for any $f(n) = \omega(1)$.
  The solutions of both Chan~\etal~\cite{cdsw2015} and
Belazzougui~\etal~\cite{BGN13,bgmnn2016} work for variable $\alpha$.
  No work has been done for dynamic range $\alpha$-minority queries.

\paragraph{Our results.} In this article we first consider the dynamic range
$\alpha$-majority problem for fixed $\alpha$ and improve the result of
Elmasry~\etal~\cite{ehmn2016} in two key performance aspects: we compress
their space requirements while reducing their time on some more general queries.
We describe a data structure
that uses even less space than Belazzougui~\etal's static representation:
$nH_k+ o(n\lg \sigma)$ bits for any $k = o(\log_{\sigma} n)$, where $\sigma$
is the alphabet size and $H_k \leq H_0 \leq \lg \sigma$ is the $k$-th order
empirical entropy of $S$. At the same time, while still supporting updates in
$O(\frac{\lg n}{\alpha})$ amortized time, we can reduce query times.
Specifically, although we still answer range $\alpha$-majority queries in
$O(\frac{\lg n}{\alpha \lg\lg n})$ time, like Elmasry~\etal, our data
structure can receive a threshold $\beta \geq \alpha$ at query time and report
the range $\beta$-majorities in $O(\frac{\lg n}{\beta \lg\lg n})$ time, rather
than $O(\frac{\lg n}{\alpha \lg\lg n})$ time. This type of queries is called
{\em range $\beta$-majority queries}. Gagie~\etal~\cite{ghmn2011} and
Chan~\etal~\cite{cdsw2015} investigated reporting $\beta$-majorities in the
static setting (i.e., variable $\alpha$) but no one has previously investigated doing so in the dynamic setting.  In summary, our time bounds are at least as good as those by Elmasry~\etal, our space bound is better to a surprising degree, and our data structure can take advantage of being given a larger threshold at query time in order to answer queries more quickly.

We also design the first solution to the dynamic range $\alpha$-minority
query problem, for fixed $\alpha$.
We can represent $S$ using $nH_k+ 2n + o(n\lg \sigma)$ bits for any $k =
o(\log_{\sigma} n)$ to answer range $\alpha$-minority queries in $O(\frac{\lg
n}{\alpha \lg\lg n})$ time, supporting symbol insertions and deletions in
$O(\frac{\lg n}{\alpha \log\log n})$ amortized time. 

As a byproduct of our main contributions, static versions of our dynamic data
structures turn out to be the first using as little as $nH_k+o(n\lg\sigma)$
bits of space ($+2n$ bits in the case of $\alpha$-minorities), for any 
$k = o(\lg_\sigma n)$. They support range 
$\alpha$-majority queries for variable $\alpha$, or $\alpha$-minority queries 
for fixed $\alpha$, in time $O((1/\alpha)\lg\lg_w\sigma)$. This time is not 
far from the $O(1/\alpha)$ achieved by Gagie~\etal~\cite{bgmnn2016} using 
$(1 +\epsilon)nH_0 + o(n)$ bits, for any constant $\epsilon > 0$, or the times 
in $(1/\alpha) \cdot \omega(1)$ they achieve within $nH_0 + o(n)(H_0 + 1)$ bits
of space. The time $O(1/\alpha)$ is optimal for $\alpha$-majority queries.

A preliminary partial version of this article appeared in {\em Proc. DCC 2017}
\cite{GHN17}. Apart from a more complete and detailed presentation, this version
includes the support for $\beta$-majority queries, the static data structure for $\alpha$-majority queries, and the dynamic data structure for $\alpha$-minority queries.

\section{Preliminaries}
\label{sec:pre}

In this section, we summarize some existing data structures that will be used in our solution. 
One such data structure is designed for the problem of maintaining a string $S$ under {\insop} and {\delop} operations to support the following operations: $\accop(i)$, which returns $S[i]$; $\rankop(c, i)$, which returns the number of occurrences of character $c$ in $S[1..i]$; and $\selop(c, i)$, which returns the position of the $i$-th occurrence of $c$ in $S$. The following lemma summarizes the currently best compressed solution to this problem, which also supports the extraction of an arbitrary substring in optimal time: 

\begin{lemma}[\cite{mn2015}]\label{lem:string}
A string of length $n$ over an alphabet of size $\sigma$ can be represented using $nH_k+ o(n\lg \sigma)$ bits for any $k = o(\log_{\sigma} n)$ to support $\accop$, $\rankop$, $\selop$, $\insop$ and $\delop$ in $O(\lg n / \lg\lg n)$ time. It also supports the extraction of a substring of length $l$ in $O(\lg n / \lg\lg n+ l/\log_{\sigma} n)$ time. 
\end{lemma}

Raman~\etal~\cite{rrr2001} considered the problem of representing a dynamic integer sequence $Q$ to support the following operations:  $\sumop(Q, i)$, which computes $\sum_{j=1}^i Q[j]$; $\searchop(Q, x)$, which returns the smallest $i$ with $\sumop(Q, i) \ge x$; and $\updateop(Q, i, \delta)$, which sets $Q[i]$ to $Q[i] + \delta$. 
One building component of their solution is a data structure for small sequences, which will also be used in our data structures: 

\begin{lemma}[\cite{rrr2001}]
\label{lem:sum}
A sequence, $Q$, of $O(\lg^{\epsilon} n)$ nonnegative integers of $O(\lg n)$ bits each, where $0 \le \epsilon < 1$, can be represented using $O(\lg^{1+\epsilon}n)$ bits to support $\sumop$, $\searchop$, and $\updateop(Q, i, \delta)$ where $|\delta| \le \lg n$, in $O(1)$ time. This data structure can be constructed in $O(\lg^{\epsilon} n)$ time, and requires a precomputed universal table occupying $O(n^{\epsilon'})$ bits for any fixed $\epsilon' > 0$. 
%The data structure requires a precomputed universal table of size $O(n^{\epsilon'})$ bits for any fixed $\epsilon' > 0$. 
%The structure can be constructed in $O(\lg^{\epsilon} n)$ time except the precomputed table. 
\end{lemma}

%We will use this lemma to encode information stored as small sequences of integers in our data structures.

\section{Compressed Dynamic Range Majority Data Structures}
\label{sec:compressed}

In this section we design compressed dynamic data structures for range $\alpha$-majority queries. 
We define three different types of queries as follows. 
Given an $\alpha$-majority query with range $[i..j]$, we compute the size, $r$, of the query range as $j - i + 1$. 
If $r \ge L$, where $L = \lceil \frac{1}{\alpha} (\lceil \frac{\lg n}{\lg\lg n} \rceil)^2 \rceil$, then we say that this query is a {\em large-sized query}. 
The query is called a {\em medium-sized query} if $L' < r < L$, where $L' = \lceil \frac{1}{\alpha} \lceil \frac{\lg n}{\lg\lg n} \rceil \rceil$. 
If $r \le L'$, then it is a {\em small-sized query}.

We represent the input sequence $S$ using Lemma~\ref{lem:string}. This supports small-sized queries immediately: By Lemma~\ref{lem:string}, we can compute the content of the subsequence $S[i..j]$, where $[i..j]$ is the query range, in $O(\frac{\lg n}{\lg\lg n} + \frac{j-i+1}{\log_{\sigma} n}) = O(\frac{\lg n}{\alpha \lg\lg n})$ time. 
We can then compute the $\alpha$-majorities in $S[i..j]$ in $O(j-i+1) = O(\frac{\lg n}{\alpha \lg\lg n})$ time using the algorithm of Misra and Gries~\cite{mg1982}. 
Thus it suffices to construct additional data structures only for large- and medium-sized queries.

\subsection{Supporting Large-Sized Range $\alpha$-Majority Queries}
\label{sec:large}

To support large-sized queries, we construct a weight-balanced B-tree~\cite{av2003} $T$ with branching parameter $8$ and leaf parameter $L$. 
We augment $T$ by adding, for each node, a pointer to the node immediately to its left at the same level, and another pointer to the node immediately to its right. 
These pointers can be maintained easily under updates, and will not affect the space cost of $T$ asymptotically. 
Each leaf of $T$ represents a contiguous subsequence, or {\em block}, of $S$, and the entire sequence $S$ can be obtained by concatenating all the blocks represented by the leaves of $T$ from left to right. 
Each internal node of $T$ then represents a block that is the concatenation of all the blocks represented by its leaf descendants.
We number the levels of $T$ by $0, 1, 2, \ldots$ from the leaf level to the root level. Thus level $a$ is higher than level $b$ if $a > b$. 
Let $v$ be a node at the $l$-th level of $T$, and let $B(v)$ denote the block it represents. Then, by the properties of weight-balanced B-trees, if $v$ is a leaf, the length of its block, denoted by $|B(v)|$, is at least $L$ and at most $2L - 1$.
If $v$ is an internal node, then $\frac{1}{2}\cdot 8^l \cdot L < |B(v)| < 2\cdot 8^l \cdot L$. 
We also have that each internal node has at least $2$ and at most $32$ children. 

We do not store the actual content of a block in the corresponding node of $T$. Instead, for each $v$, we store the size of the block that it represents, and in addition, compute and store information in a structure $C(v)$ called {\em candidate list} about symbols that can possibly be the $\alpha$-majorities of subsequences that meet certain conditions. More precisely, let $l$ be the level of $v$, $u$ be the parent of $v$, and $SB(v)$ be the concatenation of the blocks represented by the node immediately to the left of $u$ at level $l+1$, the node $u$, and the node immediately to the right of $u$ at level $l+1$. Then $C(v)$ contains each symbol that appears more than $\alpha b_l$ times in $SB(v)$, where $b_l = \frac{1}{2}\cdot 8^l \cdot L$ is the minimum size of a block at level $l$.
Since the maximum length of each block at level $l+1$ is $4b_{l+1} = 32 b_l$, we have $|SB(v)| \le 96 b_l$, and thus $|C(v)| = O(1/\alpha)$. 
To show the idea behind the candidate lists, we say that two subsequences {\em touch} each other if their corresponding sets of indices in $S$ are not disjoint. 
We then observe that, since the size of any block at level $l+1$ is greater than $8b_l$, any subsequence $S[i..j]$ touching $B(v)$ is completely contained in $SB(v)$ if $r = j-i+1$ is within $(b_l, 8b_l)$.
Since each $\alpha$-majority in $S[i..j]$ appears at least $\alpha r > \alpha b_l$ times, it is also contained in $C(v)$.
Therefore, to find the $\alpha$-majority in $S[i..j]$, it suffices to verify whether each element in $C(v)$ is indeed an answer; more details are to be given in our query algorithm later. 

Even though it only requires $O(|SB(v)|)$ time to construct $C(v)$~\cite{mg1982}, it would be costly to reconstruct it every time an update operation is performed on $SB(v)$.
To make the cost of maintaining $C(v)$ acceptable, we only rebuild it periodically by adopting a strategy by Karpinski and Nekrich~\cite{kn2008}.
More precisely, when we construct $C(v)$, we store symbols that occur more than $\alpha b_l/2$ times in $SB(v)$. We also keep a counter $U(v)$ that we increment whenever we perform {\insop} or {\delop} in $SB(v)$. 
Only when $U(v)$ reaches $\alpha b_l / 2$ do we reconstruct $C_B$, and then we reset $U(v)$ to $0$.
Since at most $\alpha b_l/2$ updates can be performed to $|SB(v)|$ between two consecutive reconstructions, any symbol that becomes an $\alpha$-majority in $|SB(v)|$ any time during these updates must have at least $\alpha b_l/2$ occurrences in $SB(v)$ before these updates are performed.
Thus we can guarantee that any symbol that appears more than $\alpha b_l$ times in $SB(v)$ is always contained in $C(v)$ during updates.
The size of $C(v)$ is still $O(b_l / \alpha)$, and, as will be shown later, it only requires $O((\lg n)/ \alpha)$ amortized time per update to $S$ to maintain all the candidate lists.

We also construct data structures to speed up a top-down traversal in $T$. These data structures are defined for the {\em marked} levels of $T$, where the $k$-th marked level is level $k\lceil (1/6) \lg\lg n\rceil$ of $T$ for $k = 0, 1, \ldots$.
Given a node $v$ at the $k$-th marked level, the number of its descendants at the $(k-1)$-st marked level is at most $32^{\lceil (1/6) \lg\lg n\rceil - 1} \le 32^{(1/6) \lg\lg n} = \lg^{5/6} n$.
Thus, the sizes of the blocks represented by these descendants, when listed from left to right, form an integer sequence, $Q(v)$, of at most $\lg^{5/6} n$ entries.
We represent $Q(v)$ using Lemma~\ref{lem:sum}, and store a sequence of pointers $P(v)$, in which $P(v)[i]$ points to the $i$-th leftmost descendant at the $(k-1)$-st marked level. 

We next prove the following key lemma regarding an arbitrary subsequence $S[i..j]$ of length greater than $L$, which will be used in our query algorithm: 

\begin{lemma}\label{lem:largecandidate}
    If $r = j-i + 1 > L$, then each $\alpha$-majority element in $S[i..j]$ is contained in $C(v)$ for any node $v$ at level $l = \lceil \frac{1}{3}\lg \frac{2r}{L} -1 \rceil$ whose block touches $S[i..j]$. 
\end{lemma}

\begin{proof}
  Let $u$ be $v$'s parent. Then $S[i..j]$ also touches $u$, and $u$ is at level $l+1$.  Let $u_1$ and $u_2$ be the nodes immediately to the left and right of $u$ at level $l+1$, respectively.

  Let $b_l$ and $b_{l+1}$ denote the minimum block size represented by nodes at level $l$ and $l+1$ of $T$, respectively. 
  Then, by the properties of weight-balanced B-trees, if $l > 0$, $b_l = \frac{1}{2} \cdot 8^l \cdot L = \frac{1}{2} \cdot 8^{\lceil \frac{1}{3}\lg \frac{2r}{L} -1 \rceil} \cdot L < \frac{1}{2} \cdot 8^{\frac{1}{3}\lg \frac{2r}{L}} \cdot L = r$. When $l = 0$, $b_l = L < r$. Thus, we always have $b_l < r$. Therefore, any $\alpha$-majority of $S[i..j]$ occurs more than $\alpha r > \alpha b_l$ times in $S[i..j]$.

  On the other hand, $b_{l+1} = \frac{1}{2} \cdot 8^{\lceil \frac{1}{3}\lg \frac{2r}{L} \rceil} \cdot L \ge \frac{1}{2} \cdot 8^{\frac{1}{3}\lg \frac{2r}{L}}\cdot L = r$. Since $S[i..j]$ touches $B(u)$, this inequality means that $S[i..j]$ is entirely contained in either the concatenation of $B(u_1)$ and $B(u)$, or the concatenation of $B(u)$ and $B(u_2)$. In either case, $S[i..j]$ is contained in $SB(v)$.
  Since any $\alpha$-majority of $S[i..j]$ occurs more than $\alpha b_l$ times in $S[i..j]$, it also occurs more than $\alpha b_l$ times in $SB(v)$. 
  As $C(v)$ includes any symbol that appears more than $\alpha b_l$ times in $SB(v)$, any $\alpha$-majority of $S[i..j]$ is contained in $C(v)$. \qed
\end{proof}

We now describe our query and update algorithms, and analyze space cost.

\begin{lemma}\label{lem:largequery}
Large-sized range $\alpha$-majority queries can be supported in $O(\frac{\lg n}{\alpha \lg\lg n})$ time. 
\end{lemma}
\begin{proof}
  Let $[i..j]$ be the query range, $r = j-i+1$ and $l = \lceil \frac{1}{3}\lg \frac{2r}{L} -1 \rceil$.
  We first look for a node $v$ at level $l$ whose block touches $S[i..j]$.
  The obvious approach is to perform a top-down traversal of $T$ to look for a node at level $l$ whose block contains position $i$. 
  During the traversal, we make use of the information about the lengths of the blocks represented by the nodes of $T$ to decide which node at the next level to descend to, and to keep track of the starting position in $S$ of the block represented by the node that is currently being visited.
  More precisely, suppose we visit node $u$ at the current level as we have determined previously that $B(u)$ contains $S[i]$. We also know that the first element in $B(u)$ is $S[p]$.
  Let $u_1, u_2, \ldots, u_d$ denote the children of $u$, where $d \le 32$. 
  To decide which child of $u$ represents a block that contains $S[i]$, we retrieve the lengths of all $|B(u_k)|$'s, and look for the smallest $q$ such that $p+\sum_{k=1}^q |B(u_k)| > i$.
  Node $u_q$ is then the node at the level below whose block contains $S[i]$, and the starting position of its block in $S$ is $p + \sum_{k=1}^{q-1} |B(u_k)|$.
  As $d \le 32$ and we store the length of the block that each node represents, these steps use constant time.

  However, if we follow the approach described in the previous paragraph, we would use $O(\lg n)$ time in total, as $T$ has $O(\lg n)$ levels.
  Thus we make use of the additional data structures stored at marked levels to speed up this process.
  If there is no marked level between the root level and $l$, then the top down traversal only descends $O(\lg \lg n)$ levels, requiring $O(\lg\lg n)$ time only.
  Otherwise, we perform the top-down traversal until we reach the highest marked level. 
  Let $x$ be the node we visit at the highest marked level.
  As $Q(x)$ stores the lengths of the blocks at the next marked level, we can perform a $\searchop$ operation in $Q(x)$ and then follow an appropriate pointer in $P(x)$ to look for the node $y$ at the second highest level that contains $S[i]$, and perform a $\sumop$ operation in $Q(x)$ to determine the starting position of $B(y)$ in $S$. These operations require constant time.
  We repeat this process until we reach the lowest marked level above level $l$, and then we descend level by level until we find node $v$.
  As there are $O(\lg n / \lg\lg n)$ marked levels, the entire process requires $O(\lg n/ \lg\lg n)$ time.

  By Lemma~\ref{lem:largecandidate}, we know that the $\alpha$-majorities of $S[i..j]$ are contained in $C(v)$. 
  We then verify, for each symbol, $c$, in $C(v)$, whether it is indeed an $\alpha$-majority by computing its number, $m$, of occurrences in $S[i..j]$ and comparing $m$ to $\alpha r$.
  As $m = \rankop(c, j) - \rankop(c, i-1)$, $m$ can be computed in $O(\lg n / \lg\lg n)$ time by Lemma~\ref{lem:string}.
  As $|C(v)| = O(1 / \alpha)$, it requires $O(\frac{\lg n}{\alpha \lg\lg n})$ time in total to find out which of these symbols should be included in the answer to the query.
  Therefore, the total query time is $O(\frac{\lg n}{\lg\lg n} + \frac{\lg n}{\alpha \lg\lg n}) = O(\frac{\lg n}{\alpha \lg\lg n})$. \qed
\end{proof}

%We have the following lemma regarding maintaining our data structures under updates:

\begin{lemma}\label{lem:largeupdate}
The data structures described in Section~\ref{sec:large} can be maintained in $O(\frac{\lg n}{\alpha})$ amortized time under update operations. 
\end{lemma}

\begin{proof}
  We show only how to support {\insop}; the support for {\delop} is similar. 

  To perform $\insop(c, i)$, we first perform a top down traversal to look for the node $v$ at level $0$ whose block contains $S[i]$. 
  During this traversal, we descend level by level as in Lemma~\ref{lem:largequery}, but we do not use the marked levels to speed up the process.
  For each node $u$ that we visit, we increment the recorded length of $B(u)$.
  In addition, we update the counters $U$ stored in the children of $u$ and in the children of the two nodes that surround $u$.
  There are a constant number of these nodes, and they can all be located in constant time by following either the edges of $T$, or the pointers between two nodes that are next to each other at the same level where we augment $T$. 

  When incrementing the counter $U$ of each node, we find out whether the candidate list of this node has to be rebuilt. To reconstruct the candidate list of a node $x$ at level $l$, we first compute the starting and ending positions of $SB(x)$ in $S$.
  This can be computed in constant time because, during the top down traversal, we have already computed the starting and ending positions of $B(v)$ in $S$, and the three nodes whose blocks form $SB(x)$, as well as the sizes of these three blocks, can be retrieved by following a constant number of pointers starting from $v$.
  We then extract the content of $SB(x)$. As $|SB(x)| \le 96 b_l$ (see discussions earlier in this section) and $b_l \ge L$, by Lemma~\ref{lem:string}, $SB(x)$ can be extracted from $S$ in $O(b_l)$ time.
  We next compute all the symbols that appear in $SB(x)$ more than $\alpha b_l / 2$ times in $O(b_l)$ time~\cite{mg1982}, and these are the elements in the reconstructed $C(x)$.
  Since the counter $U(x)$ has to reach $\alpha b_l / 2$ before $C(x)$ has to be rebuilt, the amortized cost per update is $O(1/\alpha)$.
  
  If $u$ is at a marked level, we perform a $\searchop$ operation in $O(1)$ time to locate the entry of $Q(u)$ that corresponds to the node at the next lower marked level whose block contains $i$, and perform an $\updateop$, again in $O(1)$ time, to increment the value stored in this entry. 
  So far we have used $O(1/\alpha)$ amortized time for each node we visit during the top-down traversal.  
  Since $T$ has $O(\lg n)$ levels, the overall cost we have calculated up to this point is $O((\lg n)/{\alpha})$ amortized time.

  When a node, $z$, at level $l$ of $T$ splits into two nodes $z_1$ and $z_2$, where $z_1$ is to the left of $z_2$, 
  we construct $C(z_1)$ and $C(z_2)$ in $O(b_l)$ time.
  In addition, for any node $y$ that is a child of $z_1$ or $z_2$, or a child of the node immediately to the left of $z_1$ or the right of $z_2$ at the same level,
  we reconstruct $C(y)$ in $O(b_l)$ time.
  As there are a constant number of such nodes, all these structures can be reconstructed in $O(b_l)$ time. 
  If $l$ is a marked level,
but it is not the lowest marked level, we also build $Q(z_1)$, $Q(z_2)$,
$P(z_1)$, and $P(z_2)$. We also have to
rebuild $P(z')$ and $Q(z')$, where $z'$ is the lowest ancestor of $z$ that is
on a marked level. All this takes $O(\lg^{5/6} n) = o(b_l)$ time.
  By the properties of  a weight-balanced B-tree, after a node at level $l$ has been split, it requires at least $\frac{1}{2} \cdot 8^l \cdot L = b_l$ insertions before it can be split again.
  Therefore, we can amortize the cost of reconstructing these data structures over the insertions between reconstructions, and each $\insop$ is thus charged with $O(1)$ amortized cost.
  As each $\insop$ may cause one node at each level of $T$ to split, the overall cost charged to an {\insop} operation is thus $O(\lg n)$. 

  Finally, update operations may cause the value of $L$ to change. For this to happen, the value of $\lceil \frac{\lg n}{\lg\lg n} \rceil$ must change, and this requires $\Omega(n)$ updates. 
  When this happens, we rebuild our data structure in $O(n\lg n)$ time: we can easily precompute the structures for each level of $T$ in linear time and there are $O(\lg n)$ levels. Thus, such rebuilding incurs $O(\lg n)$ amortized time for each update.
  To summarize, {\insop} can be supported in $O((\lg n)/{\alpha})$ amortized time.\qed
\end{proof}

%We now analyze the space cost of our data structures.

\begin{lemma}\label{lem:largespace}
The data structures described in Section~\ref{sec:large} occupy $o(n\lg \sigma)$ bits. 
\end{lemma}
\begin{proof}
  As $T$ has $O(n/L)$ nodes, the structure of $T$, pointers between nodes at the same level, as well as counters and block lengths stored with the nodes, occupy $O(n/L \times \lg n) = O(\frac{\alpha n(\lg\lg n)^2}{\lg n})$ bits in total.
  Each candidate list can be stored in $O((\lg\sigma)/ \alpha)$ bits, so the candidate lists stored in all the nodes use $O(n/L \times (\lg\sigma)/\alpha) = O(\frac{n\lg\sigma(\lg\lg n)^2}{\lg^2 n})$ bits in total.
  %The total size of $Q(v)$ and $P(v)$ is dominated by the cost of storing all $Q(v)$'s and $P(v)$'s at the first marked level. At this level, the total number of entires stored in $Q(v)$ is the number of nodes at level $0$ of $T$, which is $O(n/L)$. 
The size of the structures $Q(v)$ and $P(v)$ can be charged to the pointed
nodes, so there are $O(n/L)$ entries to store.
As each entry of $Q(v)$ uses $O(\lg n)$ bits, all the $Q(v)$s 
%constructed for the first marked level 
occupy $O(n/L \times \lg n) = O(\frac{\alpha n(\lg\lg n)^2}{\lg n})$ bits. The same analysis applies to $P(v)$.
Therefore, the data structures described in this section use $O(\frac{\alpha n(\lg\lg n)^2}{\lg n} + \frac{n\lg\sigma(\lg\lg n)^2}{\lg^2 n}) = o(n\lg\sigma)$ bits. 
  \qed
\end{proof}
     
\subsection{Supporting Medium-Sized Range $\alpha$-Majority Queries}
\label{sec:medium}

We could use the same structures designed in Section~\ref{sec:large} to support medium-sized queries if we simply set the leaf parameter of $T$ to be $L'$ instead of $L$, but then the resulting data structures would not be succinct.
To save space, we build a data structure $D(v)$  for each leaf node $v$ of $T$.
Our idea for supporting medium-sized queries is similar to that for large-sized queries, but since the block represented by a leaf node of $T$ is small, we are able to simplify the idea and the data structures in Section~\ref{sec:large}.
Such simplifications allow us to maintain a multi-level decomposition of $B(v)$ in a hierarchy of lists instead of in a tree, which are further laid out in one contiguous chunk of memory for each leaf node of $T$, to avoid using too much space for pointers. 

We now describe this multi-level decomposition of $B(v)$, which will be used to define the data structure components of $D(v)$.
As we define one set of data structure components in $D(v)$ for each level of this decomposition, we use $D(v)$ to refer to both the data structure that we build for $B(v)$ and the decomposition of $B(v)$. 
To distinguish a level of $D(v)$ from a level of $T$, we number each level of $D(v)$ using a non-positive integer.
At level $-l$, for $l = 0, 1, 2, \ldots, \lceil \lg(L/L') - 1\rceil$, $B(v)$ is partitioned into {\em miniblocks} of length between $L/2^l$ and $L/2^{l-1}$. 
Note that the level $0$ decomposition contains simply one miniblock, which is $B(v)$ itself, as the length of any leaf block in $T$ is between $L$ and $2L$ already.
We define $m_l = L/2^l$, which is the minimum length of a miniblock at level $-l$. 
As $L' < m_{\lceil \lg(L/L') - 1\rceil} \le 2L'$, the minimum length of a miniblock at the lowest level, i.e., level $-\lceil \lg(L/L') - 1\rceil$, is between $L'$ and $2L'$. 

For each miniblock $M$ at level $-l$ of $D(v)$, we define its {\em predecessor}, $\predop(M)$, as follows: If $M$ is not the leftmost miniblock at level $-l$ of $D(v)$, then $\predop(M)$ is the miniblock immediately to its left at the same level.
Otherwise, if $v$ is not the leftmost leaf ($\predop(M)$ is null otherwise), let $v_1$ be the leaf immediately to the left of $v$ in $T$, and $\predop(M)$ is defined to be the rightmost miniblock at level $-l$ of $D(v_1)$.
Similarly, we define the {\em successor}, $\succop(M)$, of $M$ as the miniblock immediately to the right of $M$ at level $-l$ of $D(v)$ if such a miniblock exists. Otherwise, $\succop(M)$ is the leftmost miniblock at level $-l$ of $D(v_2)$ where $v_2$ is the leaf immediately to the right of $v$ in $T$ if $v_2$ exists, or null otherwise.
Then, the candidate list, $C(M)$, of $M$ contains each symbol that occurs more than $\alpha m_l / 2$ times in the concatenation of $M$, $\predop(M)$ and $\succop(M)$.
To maintain $C(M)$ during updates, we use the same strategy in Section~\ref{sec:large} that is used to maintain $C(v)$.
More specifically, we store a counter $U(M)$ so that we can rebuild $C(M)$ after exactly $\alpha m_l / 4$ update operations have been performed to $M$, $\predop(M)$ and $\succop(M)$.
Whenever we perform the reconstruction, we include in $C(M)$ each symbol that occurs more than $\alpha m_l / 4$ times in the concatenation of $M$, $\predop(M)$ and $\succop(M)$.
Since $|\predop(M)| + |M| + |\succop(M)| \le 6 m_l$, the number of symbols included in $C(M)$ is at most $24 / \alpha$.

The precomputed information for each miniblock $M$ includes $|M|$, $C(M)$, and $U(M)$.
These data for miniblocks at the same level, $-l$, of $D(v)$ are chained together in a doubly linked list $L_l(v)$. 
$D(v)$ then contains these $O(\lg(L/L')) = O(\lg\lg n)$ lists. 
We cannot, however, afford storing each list in the standard way using pointers of $O(\lg n)$ bits each, as this would use too much space.
Instead, we lay them out in a contiguous chunk of memory as follows:
We first observe that the number of miniblocks at level $-l$ of $D(v)$ is less than $2L/ (L/2^l) = 2^{l+1}$.
Thus, the total number of miniblocks across all levels is less than $2 \cdot 2^{\lceil \lg(L/L') - 1\rceil + 1} - 1 < 4L/L'$.
We then use an array $A(v)$ of $\lceil 4L/L' \rceil$ fixed-size {\em slots} to store $D(v)$, and each slot stores the precomputed information of a miniblock.

To determine the size of a slot, we compute the maximum number of bits needed to encode the precomputed information for each miniblock $M$. 
$C(M)$  can be stored in $\lceil 24 /\alpha \rceil \cdot \lceil \lg \sigma\rceil$  bits.
As $M$ has less than $2L$ elements, its length can be encoded in $\lceil \lg (2L) \rceil$ bits.
The counter $U(M)$ can be encoded in $\lceil \lg (\alpha m_l / 4) \rceil < \lceil \lg (\alpha L / 2) \rceil \le \lceil \lg (L / 2) \rceil  $ bits.
The two pointers to the neighbours of $M$ in the linked list can be encoded as the indices of these miniblocks in the memory chunk.
Since there are $\lceil 4L/L' \rceil$ slots, each pointer can be encoded in $\lceil \lg \lceil 4L/L' \rceil\rceil$ bits.
Therefore, we set the size of each slot to be $\lceil 24 /\alpha \rceil \cdot \lceil \lg \sigma\rceil + 2 \lceil \lg L \rceil + 2\lceil \lg \lceil 4L/L' \rceil\rceil$ bits. 

We prepend this memory chunk with a header. This header encodes the indices of the slots that store the head of each $L_l(v)$. 
As there are $\lceil \lg(L/L')\rceil$ levels and each index can be encoded in $\lceil \lg \lceil 4L/L' \rceil\rceil$ bits, the header uses $\lceil \lg(L/L')\rceil \cdot \lceil \lg \lceil 4L/L' \rceil\rceil$ bits.
Clearly our memory management scheme allows us to traverse each doubly linked list $L_l(v)$ easily.
When miniblocks merge or split during updates, we need to perform insertions and deletions in the doubly linked lists.
To facilitate these updates, we always store the precomputed information for all miniblocks in $D(v)$ in a prefix of $A(v)$, and keep track of the number of used slots of $A(v)$. 
When we perform an insertion into a list $L_l(v)$, we use the first unused slot of $A$ to store the new information, and update the header if the newly inserted list element becomes the head.
When we perform a deletion, we copy the content of the last used slot (let $M'$ be the miniblock that corresponds to it) into the slot corresponding to the deleted element of $L_l(v)$. We also follow the pointers encoded in the slot for $M'$ to locate the neighbours of $M'$ in its doubly linked list, and update pointers in these neighbours that point to $M'$. 
If $M'$ is the head of a doubly linked list (we can determine which list it is using $|M'|$), we update the header as well.
The following lemma shows that our memory management strategy does, indeed, save space:

\begin{lemma}\label{lem:mediumspace}
  The data structures described in Section~\ref{sec:medium} occupy $o(n\lg \sigma)$ bits. 
\end{lemma}

\begin{proof}
  We first analyze the size of the memory chunk  storing $D(v)$ for each leaf $v$ of $T$.
  By our analysis in previous paragraphs, we observe that the header of this chunk uses $O((\lg\lg n)^2)$ bits.
  Each slot of $A(v)$ uses $O(\frac{\lg\sigma}{\alpha} + \lg\lg n)$ bits, and $A(v)$ has $O(\lg n / \lg\lg n)$ entries.
  Therefore, $A(v)$ occupies $O(\frac{\lg\sigma\lg n}{\alpha\lg\lg n} + \lg n)$ bits.
  Hence the total size of the memory chunk of each leaf of $T$ is $O(\frac{\lg\sigma\lg n}{\alpha\lg\lg n} +\lg n)$ bits.
  As there are $O(n/L)$ leaves in $T$, the data structures described in this section use $O(\frac{n\lg\sigma\lg\lg n}{\lg n} + \frac{\alpha n(\lg\lg n)^2}{\lg n}) = o(n\lg\sigma)$ bits. 
\qed
\end{proof}

We now show how to support query and update operations.

\begin{lemma}\label{lem:mediumquery}
Medium-sized range $\alpha$-majority queries can be supported in $O(\frac{\lg n}{\alpha \lg\lg n})$ time. 
\end{lemma}

\begin{proof}
  Let $[i..j]$ be the query range and let $r = j-i+1$.  We first perform a top down traversal in $T$ to locate the leaf, $v$, that represents a block containing $S[i]$ in $O(\frac{\lg n}{\lg\lg n})$ time using the approach described in the proof of Lemma~\ref{lem:largequery}.
  In this process, we can also find the starting position of $B(v)$ in $S$. 

  We next make use of $D(v)$ to answer the query as follows. Let $l = \lceil \lg(L/r) - 1\rceil$.
  As $m_l = L/2^{\lceil \lg(L/r) - 1\rceil}$, we have $m_l / 2 \le r < m_l$.
  We then scan the list $L_l(v)$ to look for a miniblock, $M$, that contains $S[i]$ at level $-l$.
  This can be done by first locating the head of $L_l(v)$ from the header of the memory chunk that stores $D(v)$, and then performing a linear scan, computing the starting position of each miniblock in $L_l(v)$ along the way.
  As $L_l(v)$ has at most $O(L/L') = O(\frac{\lg n}{\lg\lg n})$ entries, we can locate $M$ in $O(\frac{\lg n}{\lg\lg n})$ time. 
  %During this process, we can also find out the starting position of $M$ in $S$, as well as the length of $M$.
  Since $m_l > r$, $S[i..j]$ is either entirely contained in the concatenation of $\predop(M)$ and $M$, or in the concatenation of $M$ and $\succop(M)$.
  Thus each $\alpha$-majority of $S[i..j]$ must occur more than $\alpha r > \alpha m_l / 2$ times in the concatenation of $\predop(M)$, $M$ and $\succop(M)$.
  Therefore, each $\alpha$-majority of $S[i..j]$ is contained in $C(M)$.
  We can then perform $\rankop$ operations in $S$ to verify whether each symbol in $C(M)$ is indeed an $\alpha$-majority of $S[i..j]$.
  As $C(M)$ has $O(1/\alpha)$ symbols, this process requires $O(\frac{\lg n}{\alpha \lg\lg n})$ time.
  The total query time is hence $O(\frac{\lg n}{\alpha \lg\lg n})$.
\qed
\end{proof}

\begin{lemma}\label{lem:mediumupdate}
The data structures described in Section~\ref{sec:medium} can be maintained in $O(\frac{\lg n}{\lg\lg n} + \frac{\lg\lg n}{\alpha})$ amortized time under update operations. 
\end{lemma}
\begin{proof}
    We show only how to support {\insop}; the support for {\delop} is similar.

    To perform $\insop(c, i)$, we first perform a top down traversal in $T$ to locate the leaf, $v$, that represents a block containing $S[i]$ in $O(\frac{\lg n}{\lg\lg n})$ time.
    We then increment the recorded lengths of all the miniblocks that contain $S[i]$. We also increment the counters $U$ of these miniblocks, as well as the counters of their predecessors and successors. 
    All the miniblocks whose counters should be incremented are located in $D(v)$, $D(v_1)$ and $D(v_2)$, where $v_1$ and $v_2$ are the leaves immediately to the left and right of $v$ in $T$.
    At each level $-l$, we scan each doubly linked list $L_l(v)$, $L_l(v_1)$ and $L_l(v_2)$ to locate these miniblocks.
    Since $D(v)$, $D(v_1)$ and $D(v_2)$ have $O(\frac{\lg n}{\lg\lg n})$ miniblocks in total over all levels, it requires $O(\frac{\lg n}{\lg\lg n})$ to find these miniblocks and update them.

    The above process can find all these miniblocks, as well as their starting and ending positions in $S$.
    It may be necessary to reconstruct the candidate list of these miniblocks. Similarly to the analysis in the proof of Lemma~\ref{lem:largeupdate}, the candidate list of each of these miniblocks can be maintained in $O(1/\alpha)$ amortized time.
    Since there are $O(\lg \lg n)$ levels in $D(v)$, $D(v_1)$ and $D(v_2)$, and only a constant number of miniblocks needing rebuilding at each level, $O((\lg\lg n) / \alpha)$ amortized time will be required to reconstruct all of them.

    An insertion may also cause a miniblock $M$ to split. As in the proof of Lemma~\ref{lem:largeupdate}, we compute the candidate lists and other required information for the miniblocks created as a result of the split in time linear in the length of $M$, and amortize the cost over the insertions that lead to the split. As the number of these insertions is also proportional to the length of $M$, the amortized cost is again $O(1)$. As there can possibly be a split at each level of $D(v)$, it requires $O(\lg\lg n)$ amortized time to handle them.
    Finally, when the value of $L'$ changes, we rebuild all the data structures designed in this section.
    Since these data structures are constructed for $O(\lg\lg n)$ levels and the structures for each level can be rebuilt in linear time, this process incurs $O(\lg\lg n)$ amortized time.
    Therefore, the total time required to support {\insop} is $O(\frac{\lg n}{\lg\lg n} + \frac{\lg\lg n}{\alpha})$. 
\qed
\end{proof}

Combining Lemma~\ref{lem:string} and Lemmas~\ref{lem:largequery}-\ref{lem:mediumupdate}, we obtain our first result, when the structure is queried for $\alpha$-majorities.

\begin{theorem} \label{thm:alpha}
  For any $0<\alpha<1$, a sequence of length $n$ over an alphabet of size $\sigma$ can be represented using $nH_k+ o(n\lg \sigma)$ bits for any $k = o(\log_{\sigma} n)$ to answer range $\alpha$-majority queries in $O(\frac{\lg n}{\alpha \lg\lg n})$ time, and to support symbol insertions and deletions in $O(\frac{\lg n}{\alpha})$ amortized time.
\end{theorem}

\section{Supporting $\beta$-Majorities}
\label{sec:beta}

Theorem~\ref{thm:alpha} supports range $\alpha$-majority queries, where $\alpha$ is chosen at construction time. We now enhance our data structure to find range $\beta$-majorities, for any $\beta \ge \alpha$ given at query time together with the interval $[i..j]$. While it is easy to answer those queries in time $O(\frac{\lg n}{\alpha \lg\lg n})$, our goal is to reach time $O(\frac{\lg n}{\beta \lg\lg n})$. Updates are still carried out in amortized time $O(\frac{\lg n}{\alpha})$.

Although we have not used this in previous sections, note that we can focus our attention in the case $\beta > 1/\sigma$, since otherwise we can directly check the range of $S$ for each of the $\sigma$ symbols $c$, reporting those where $\rankop(c, j) - \rankop(c, i-1) > \beta r$, all in time $O(\frac{\sigma \log n}{\log\log n}) = O(\frac{\lg n}{\beta \lg\lg n})$. Thus, at construction time we can set $\alpha$ to $1/\sigma$ if $\alpha$ turns out to be smaller. This implies, in particular, that all our $\frac{1}{\alpha}$ in the complexities can be replaced by $\min(\frac{1}{\alpha},\sigma)$. We will also use the fact that $\log\frac{1}{\alpha} = O(\log\sigma) \cap O(\log n)$. Similarly, it makes sense to consider $\alpha \le 1/2$ only, as otherwise we use the solution for $\alpha=1/2$ and report only the true $\alpha$-majorities found, within the same complexity.

\subsection{Large and Medium-Sized Intervals}

For large and medium-sized intervals, it is not difficult to answer $\beta$-majority queries within the desired time. Note that, in those cases, the crux of the solution is to verify a list of candidates, $C(v)$ in the block $v$ (for large intervals) or $C(M)$ in the miniblock $M$ (for medium-sized intervals), both of size $O(1/\alpha)$. It is sufficient that those lists are sorted by decreasing frequency of the elements and that we stop verifying them when reaching an element with frequency below $\beta r$. Since $r > b_l$ in large intervals and $r \ge m_l/2$ in medium-sized intervals, there can be only $O(1/\beta)$ such candidates and we solve the query in time $O(\frac{\log n}{\beta \log\log n})$.

We can maintain those list only approximately sorted, however. When the lists are created, we sort them by decreasing frequency.
%according to $\lceil \log (1/f(x)) \rceil$, where $f(x)$ is the relative frequency of element $x$ in the range. This produces only $O(\log n)$ different values, so radix sort orders the elements in linear time (we sort $\Theta(L')$ to $\Theta(L)$ elements, so a constant number of passes of radix sort with universe size $\sqrt{\log n}$ achieves linear time). Since our ordering has an imprecision factor of at most $1/2$, we will stop verifying only when the frequency drops below $\beta r/2$. This does not change the time complexity. In addition, 
However, the elements can later change their frequency upon updates, but only by a maximum of $\gamma=\alpha b_l/2$ (for large intervals) or $\gamma=\alpha m_l/4$ (for medium-sized intervals), before we rebuild the lists. We do not store symbol frequencies, just their order. This ordering is not modified upon updates, only when the lists are rebuilt. Therefore, we can stop verifying safely only when the frequency we compute on the fly drops below $\beta r-\gamma$, since this guarantees than the next element cannot have a current frequency over $\beta r$. Since $r > b_l$ for large intervals and $r \ge m_l/2$ for medium-sized intervals, it holds that $\beta r - \gamma > \beta b_l/2$ for large intervals and $\beta r - \gamma \ge \beta m_l/4$ for medium-sized intervals, and therefore the resulting complexity is in both cases $O(\frac{\log n}{\beta \log\log n})$.

\subsection{Small Intervals}

The small ranges, which were solved by brute force, pose a more difficult problem, because now we cannot afford scanning a block of $S$ of size $O(\frac{\log n}{\alpha \log\log n})$. To handle small ranges, we add further structures to our tree leaves, which contain $L'$ to $2L'$ elements for $L' = \lceil \frac{1}{\alpha} \lceil \frac{\lg n}{\lg\lg n} \rceil \rceil$. The leaves will be further partitioned into halves repeatedly in $\log(1/\alpha)$ levels, until reaching size between $L^*$ and $2L^*$, for $L^* = \lceil \frac{\lg n}{\lg\lg n} \rceil$.

These additional levels, numbered $-l^*$ for $l = 0,\ldots,\lfloor \lg(L'/L^*) \rfloor$, are organized much as the miniblocks of Section~\ref{sec:medium}. Indeed, our highest level, $-0^*$, is the same level, $-\lceil \lg(L/L') \rceil$, as the deepest one of Section~\ref{sec:medium}. The main difference is that, in the new levels, not only the sizes $m_{l^*}$ are halved as we descend, but also the majority thresholds are doubled: we use the value $\alpha_{l^*} = \alpha \cdot 2^{l^*}$ to define the candidate lists $C(M)$ at level $-l^*$. In our last level, $-l^* = -\lfloor \lg(L'/L^*) \rfloor$, it holds that $\alpha_{l^*} = \Theta(1)$ (precisely, $\alpha_{l^*} > 1/2$) and $m_{l^*} = O(\frac{\log n}{\log\log n})$ (precisely, $L^* \le m_{l^*} \le 2L^*$). 

%The list of candidates $C(M)$ of a miniblock of level $-l^*$ stores all the elements that occurs more than $\alpha_{l^*} m_{l^*}/4$ times in the area $\predop(M) \cdot M \cdot \succop(M)$, of length at most $6m_{l^*}$. Thus there can be at most $24/\alpha_{l^*}$. TOO REPETITIVE

Because, at each level $-l^*$, $C(M)$ can store at most $24/\alpha_{l^*} = 24/(\alpha\, 2^{l^*})$ elements, we do not store together the information of all the miniblocks descending from a leaf block, as done in Section~\ref{sec:medium}. Rather, we stratify it per level $-l^*$. For each leaf block, we have an array of $O(\log\frac{1}{\alpha})$ entries, one per level $-l^*$, to memory areas of miniblocks of that level descending from the leaf block. Within each memory area, the slots are of the same size, as in Section~\ref{sec:medium}. 

We also impose further structure to the linked lists of miniblocks inside each memory area: the list nodes are not anymore linked, but they are the leaves of a B-tree of arity $B$ to $2B$, for $B = \sqrt{\lg n}$. Since the list at level $-l^*$ has $2^{l^*} < \frac{1}{\alpha}+1$ elements, the B-tree is of height $O(\log(1/\alpha) / \log\log n)$. Each B-tree node stores the up to $2\sqrt{\log n}$ subtree sizes (measured in terms of number of positions of $S$ stored in all the subtree leaves) using Lemma~\ref{lem:sum}, which allows routing the search for a given position in $S$ in constant time per B-tree node. To facilitate memory management, we have one memory area for the B-tree nodes and another for the list nodes, so that memory slots are of the same size within each area.

Finally, we use a new arrangement to store the lists $C(M)$ in these miniblocks. Instead of representing the candidate symbols directly, we store one position of $\predop(M) \cdot M \cdot \succop(M)$ where the symbol appears. The actual symbol can then be obtained with an access to $S$ in time $O(\frac{\lg n}{\lg\lg n})$. Further, we sort all the $O(1/\alpha_{l^*})$ positions of the candidates as follows: The primary criterion for the sort is $\lceil \log(1/f) \rceil$, where $f$ is the relative frequency of the element in $\predop(M) \cdot M \cdot \succop(M)$. The secondary criterion, when the first produces ties, is the increasing order of the positions in $\predop(M) \cdot M \cdot \succop(M)$ we use to represent the symbols.

Therefore, the list $C(M)$ is partitioned into $O(\log n)$ {\em chunks} of
symbols with the same quantized frequency, $q_f = \lceil \log(1/f) \rceil$,
and the positions stored are increasing within each chunk. Those chunks are
then represented as the differences between consecutive positions using
$\gamma$-codes \cite{BCW90}, and a difference of zero is used to signal the
end of a chunk. By Jensen's Inequality, the number of bits required to
represent $k$ differences that add up to $m$ is $O(k\log(m/k))$.%
\footnote{Since $\gamma$-codes can only represent positive numbers and we want
to use zero to signal end of chunks, we will always use the code for $x+1$
to represent the number $x$. This adds only $O(k)$ extra bits.}
 Since there are at most $2^{q_f}$ elements with quantized frequency $q_f$, their chunk is represented with $O(2^{q_f}(\log(m)-q_f))$ bits. Adding up to relative frequency $f^* = \alpha_{l^*}/24$ (i.e., the minimum for a candidate stored in $C(M)$), the total space adds up to the order of
$$ \sum_{q_f=0}^{q_f=\log\lceil 24/\alpha_{l^*}\rceil} 2^{q_f}(\log m -q_f)
~~=~~
2\left\lceil\frac{24}{\alpha_{l^*}}\right\rceil \left(\log m - \log \left\lceil\frac{24}{\alpha_{l^*}}\right\rceil+1\right)-3 = O\left(\frac{\log (\alpha_{l^*} m)}{\alpha_{l^*}}\right),$$
and since in our case $m=O(m_{l^*})=O(\frac{\lg n}{\alpha 2^{l^*}\lg\lg n}) =
O(\frac{\lg n}{\alpha_{l^*} \lg\lg n})$, the total space to represent $C(M)$ is $O((1/\alpha_{l^*})\lg\lg n)$.

\begin{lemma}\label{lem:betaspace}
  The data structures described in Section~\ref{sec:beta} occupy $o(n\lg \sigma)$ bits. 
\end{lemma}
\begin{proof}
The analysis is analogous to that of Lemma~\ref{lem:mediumspace}. The number of miniblocks at level $-l^*$ of each array $A(v)$ is at most $2L/(L'/2^{l^*}) = O(\frac{\lg n}{\lg\lg n}\cdot 2^{l^*+1})$. The size of the miniblocks includes the space to store the list of candidates, $O((1/\alpha_{l^*})\log\log n)$, plus a constant number of $(\log L)$-bit counters and pointers, which require $O(\log\log n + \log\frac{1}{\alpha})$ further bits. The B-tree nodes, stored in another memory area, require $O(B \lg L)$ bits per node, but have $O(1/B)$ nodes per miniblock $M$, thus their space is already covered in our formula. All this adds up to $O((1/\alpha_{l^*})\log\log n + \log\frac{1}{\alpha})$ bits per miniblock, which multiplied by the number of miniblocks at level $-l^*$ of $A(v)$ yields $O(\frac{\log n}{\alpha} + \frac{\lg n \lg\frac{1}{\alpha}}{\lg\lg n}\cdot 2^{l^*})$ bits. 
Summing up this space over all the $\lg\frac{1}{\alpha}$ levels $-l^*$, we obtain $O(\frac{\lg n\lg\frac{1}{\alpha}}{\alpha} + \frac{\lg n \lg\frac{1}{\alpha}}{\alpha\lg\lg n})= O(\frac{\lg n\lg\frac{1}{\alpha}}{\alpha})$ bits. Finally, multiplying this space by the $O(n/L)$ leaves, we obtain $O(\frac{n\lg\frac{1}{\alpha} (\lg\lg n)^2}{\lg n}) = o(n\lg\sigma)$ bits.

We also have $O(\log\frac{1}{\alpha})$ global pointers for the memory areas of each level $-l^*$, which multiplied by the $O(n/L)$ leaves yields $O(\frac{\alpha n\log\sigma (\lg\lg n)^2}{\lg n})=o(n\log\sigma)$ bits in total.
\qed
\end{proof}

Now we show how to support range $\beta$-majority queries with this structure.

\begin{lemma}\label{lem:betaquery}
Small-sized range $\beta$-majority queries, for any $\beta \ge \alpha$, can be supported in $O(\frac{\lg n}{\beta \lg\lg n})$ time.
\end{lemma}
\begin{proof}
After we arrive at the corresponding leaf block $v$ in time $O(\frac{\lg n}{\lg\lg n})$ as in the proof of Lemma~\ref{lem:mediumquery}, we choose the level $-l^*$ according to $r=j-i+1$: it must hold that $m_{l^*}/2 \le r < m_{l^*}$, i.e., $\frac{\lg n}{2r \lg\lg n} \le \alpha_{l^*} < \frac{\lg n}{r \lg\lg n}$. This level is appropriate to apply the same reasoning of Lemma~\ref{lem:mediumquery}, and it exists whenever $2L^* \le r < L'$. On the other hand, we need that $\alpha_{l^*} \le \beta$ in order to ensure that the candidates stored in $C(M)$ (which include all the possible $\alpha_{l^*}$-majorities) include all the possible $\beta$-majorities. Since $\alpha_{l^*} < \frac{\lg n}{r \lg\lg n}$, it suffices that $\frac{\lg n}{r \lg\lg n} \le \beta$ to have $\alpha_{l^*} \le \beta$. This condition is equivalent to $r \ge \frac{\lg n}{\beta\lg\lg n}$. We can always assume this condition to be true, since otherwise we can use Lemma~\ref{lem:string} to extract $S[i..j]$ and find its majorities in time $O(\frac{\lg n}{\beta\lg\lg n})$ without the help of any other data structure.

We then traverse the B-tree of level $-l^*$ so as to find the appropriate miniblock $M$. The traversal takes time $O(\frac{\lg(1/\alpha)}{\lg\lg n}) = O(\frac{\lg n}{\lg\lg n})$. Once we arrive at the proper miniblock $M$, we scan the successive chunks of $C(M)$. Since the frequencies in the next chunk (at the moment of list construction) could not be more than those in the current chunk, and since some frequency may have increased by at most $\alpha_{l^*} m_{l^*}/4$ since the last reconstruction, we can safely stop the scan when the highest frequency seen in the current chunk  does not exceed $\beta r - \alpha_{l^*} m_{l^*}/4$. 

Let us analyze the cost we incur to scan up to this threshold. Since frequencies can also decrease by up to $m_{l^*}/4$ until the next reconstruction, an element with current frequency over $\beta r - \alpha_{l^*} m_{l^*}/4$ must have had frequency over $\beta r - \alpha_{l^*} m_{l^*}/2 \ge (\beta-\alpha_{l^*}) m_{l^*}/2$ when the list was built, and thus its relative frequency was $f \ge (\beta-\alpha_{l^*})/12$. Its quantized frequency was therefore $q_f \le \lceil \log(12/(\beta-\alpha_{l^*})) \rceil$, and thus we might have to process up to $2^{q_f+1} = O(\frac{1}{\beta-\alpha_{l^*}})$ elements before covering its chunk. This is $O(1/\beta)$ if $\beta \ge 2\alpha_{l^*}$; otherwise we might use the argument that the whole list is of size $O(1/\alpha_{l^*}) = O(1/\beta)$ anyway. Therefore, we try out each candidate using $\rankop$ on $S$ in time $O(\frac{\lg n}{\beta \lg\lg n})$ and complete the query.
\qed
\end{proof}

Finally, we show that we can still maintain the structure within the original time.

\begin{lemma}\label{lem:betaupdate}
The data structures described in Section~\ref{sec:beta} can be maintained in amortized time $O(\frac{\lg^2(1/\alpha)}{\log\log n}+\frac{1}{\alpha}+\frac{\lg n}{\lg\lg n})$ under update operations.
\end{lemma}
\begin{proof}
For large and medium blocks, the only difference is that we must sort the candidate lists $C(v)$ and $C(M)$ by decreasing frequency. This is not difficult because we already spend time $O(b_l)$ (for large ranges) and $O(m_l)$ (for medium-sized ranges) in building the lists. The frequencies range over a universe of the same size, thus we can sort them within the same times, $O(b_l)$ or $O(m_l)$, using radix sort.

The maintenance procedure for the levels $-l^*$ is very similar to that of miniblocks described in Lemma~\ref{lem:mediumupdate}. One difference is that, when changes in a miniblock $M$ occurs, we must update its size upwards in its B-tree. This adds $O(\frac{\log(1/\alpha)}{\log\log n})$ time, because Lemma~\ref{lem:sum} allows us update each B-tree node in constant time. Node splits and merges require $O(B)$ time, but these amortize to $O(1)$. Finally, a single update requires modifying the B-trees in all the $\log(1/\alpha)$ levels $-l^*$, for a total update cost of $O(\frac{\log^2(1/\alpha)}{\log\log n})$.

The amortized cost to reconstruct a list $C(M)$ at level $-l^*$ is $O(1/\alpha_{l^*})$, including the special sorting and encoding we use. Since a single update is reflected in all the levels, we must add up this cost for all $-l^*$, yielding $O(1/\alpha)$. Updates also need time $O(\frac{\log n}{\lg\lg n})$ to reach the desired miniblock.
\qed
\end{proof}

We now have all the elements to prove our main result. Note that $O(\frac{\lg n}{\alpha})$ encompasses all the update costs for the three range sizes.

\begin{theorem} \label{thm:beta}
  For any $0<\alpha<1$, a sequence of length $n$ over an alphabet of size $\sigma$ can be represented using $nH_k+ o(n\lg \sigma)$ bits for any $k = o(\log_{\sigma} n)$ to answer range $\beta$-majority queries for any $\beta \ge \alpha$ in $O(\frac{\lg n}{\beta \lg\lg n})$ time, and to support symbol insertions and deletions in $O(\frac{\lg n}{\alpha})$ amortized time.
\end{theorem}

\subsection{A Static Variant}

Since none of the current static data structures supporting range 
$\alpha$-majority queries reaches the high-order entropy space we obtain in 
Theorem~\ref{thm:beta}, we now consider a static version of our data structure.

A static variant of our solutions uses blocks and miniblocks of fixed size, so 
one can access in constant time the desired block at the corresponding level 
$l$, $−l$, or $-l^*$, and then try out the prefix of $O(1/\beta)$ stored 
candidates that covers all the possible $\beta$-majorities. All that 
precomputed data amounts to $o(n\lg\sigma)$ bits of space, even when built for 
the minimum $\alpha$ of interest, $\max(1/n,1/\sigma)$, as shown in 
Lemmas~\ref{lem:largespace}, \ref{lem:mediumspace}, and \ref{lem:betaspace}. 
Using a sequence representation that uses $nH_k + o(n\lg\sigma)$ bits 
\cite[Thm. 11]{BN15} and answers access queries in time $O(1)$ and rank 
queries in time $O(\lg\lg_w\sigma)$, where $w = \Omega(\lg n)$ is the RAM word 
size in bits, we can solve $\beta$-majority queries in time 
$O((1/\beta) \lg\lg_w\sigma)$.

Since the structure has $O(\lg n)$ levels and each level is built in linear 
time as described in Lemmas~\ref{lem:largeupdate}, \ref{lem:mediumupdate}, and
\ref{lem:betaupdate}, the construction time is $O(n\lg n)$. The sequence
representation we use \cite[Thm. 11]{BN15} is built in linear time.

Interestingly, since update times are irrelevant in this case, the asymptotic
time and space complexities of our data structure do not depend on $\alpha$. 
Thus, our structure can be built directly for the minimum relevant value of 
$\alpha$, $\max(1/n,1/\sigma)$, and then it can be queried for any value of 
$\beta$ (if $\beta \le \max(1/n,1/\sigma)$, we just try out all the $\sigma$ 
possible candidates using $\rankop$ on $S[l..r]$).
Thus, this data structure can be used to answer range $\alpha$-majorities for variable $\alpha$, which is even more powerful than the range $\beta$-majority query. 
The following theorem presents our result, which is stated as a solution to the range $\alpha$-majority problem for variable $\alpha$. 

\begin{theorem}
On a RAM machine of $w = \Omega(\lg n)$ bits, a sequence of length n over an 
alphabet of size $\sigma$ can be represented using $nH_k + o(n\lg\sigma)$ bits 
for any $k = o(\lg_\sigma n)$ to answer range $\alpha$-majority queries for any 
$0 < \alpha < 1$ defined at query time, in $O((1/\alpha) \lg\lg_w\sigma)$
time. The structure is built in $O(n\lg n)$ time.
\end{theorem}

%\begin{theorem}
%On a RAM machine of $w = \Omega(\lg n)$ bits, a sequence of length n over an 
%alphabet of size $\sigma$ can be represented using $nH_k + o(n\lg\sigma)$ bits 
%for any $k = o(\lg_\sigma n)$ to answer range $\beta$-majority queries for any 
%$0 < \beta < 1$ defined at query time, in $O((1/\beta) \lg\lg_w\sigma)$
%time. The structure is built in $O(n\lg n)$ time.
%\end{theorem}

\section{Finding $\alpha$-Minorities}

We now introduce the first dynamic structure to find $\alpha$-minorities
in array ranges. We build on the idea of Chan~\etal~\cite{cdsw2015}, who find 
$A=1+\lfloor 1/\alpha\rfloor$ distinct elements in $S[l..r]$ and try them out
one by one, since one of those must be a minority (there may be no minority if
there are less than $A$ distinct elements in $S[l..r]$). A succinct static 
structure based on this idea \cite{BGN13} uses an $O(n)$-bit range minimum 
query data structure \cite{FH11}, of which no dynamic succinct version exists.

We use a different dynamic arrangement that can be implemented in succinct 
space. We partition $S$ into {\em pieces}, which contain $A$ to $3A$ {\em 
distinct} elements, except when $S$ contains a single piece with less than $A$
distinct elements. The following property is the key to find an 
$\alpha$-minority in time $O(\frac{\log n}{\alpha \log\log n})$.

\begin{lemma}
If $S[l..r]$ overlaps one or two pieces only and it has an $\alpha$-minority, 
then this minority element is one of the distinct elements in those pieces. If $S[l..r]$ contains
a piece that is not the last one, then one of the distinct elements in that
contained piece is a minority in $S[l..r]$.
\end{lemma}
\begin{proof}
If $S[l..r]$ overlaps one or two pieces only, then all of its distinct elements
are also distinct elements in some of those overlapped pieces, so the result 
holds. If $S[l..r]$ contains a piece with $A$ distinct elements, then one of
those must be an $\alpha$-minority of $S[l..r]$, since not all of them can 
occur more than $\alpha \cdot (j-i+1)$ times in $S[l..r]$.
\qed
\end{proof}

Our data structure is formed by a compressed dynamic representation of $S$
using $nH_k+o(n\log\sigma)$ bits (Lemma~\ref{lem:string}) plus two dynamic
bitvectors that add $2n+o(n)$ bits:
\begin{enumerate}
\item $P[1..n]$, where $P[i]=1$ iff a new piece starts at $S[i]$.
\item $C[1..n]$, where each distinct element in each piece has one arbitrary
occurrence at position $j$ (within the piece) marked with $C[j]=1$.
\end{enumerate}
The dynamic bitvectors support the operations $\accop$, $\rankop$, $\selop$,
$\insop$, and $\delop$, in time $O(\log n /\log\log n)$ (see 
\cite[Lem.~8.1]{NS14} or \cite{HM10}).

To find an $\alpha$-minority in $S[l..r]$, we use $\rankop$ and $\selop$ on
$P$ to determine the first and the last piece overlapped by $[l..r]$.
More precisely, we compute the starting position of the first of these pieces as $x=\selop_1(P,\rankop_1(P,l))$ and the ending position of the last as $y=\selop_1(P,\rankop_1(P,r)+1)-1$.
If there
are one or two pieces overlapped by $[l..r]$, that is, $\rankop_1(P,y)-\rankop_1(P,x-1) \le 2$, we 
try out all their at most $6A$ distinct elements as follows:
For $k=1,2,\ldots$, we find their $k$-th distinct element, $c$, in $S[l..r]$ using the formula $c = S[p]$ for $p=\selop_1(C,\rankop_1(C,x-1)+k)]$.
We then compute $\rankop_c(S,r)-\rankop_c(S,l-1)$ to count how many times $c$ occurs in 
$S[l..r]$ and thus determine whether $c$ is an $\alpha$-minority.
We repeat this process until we find and return an $\alpha$-minority, or until
$p > y$, in which case we report that there is no $\alpha$-majority in the query range. 
%, using $\rankop$ and $\selop$ on 
%$C$ to find them ($c = S[p]$ for $p=\selop_1(C,\rankop_1(C,x-1)+k)]$, 
%$k=1,2,\ldots$, stopping when $p>y$), 
%and $\rankop_c(S,r)-\rankop_c(S,l-1)$ to count how many times they occur in 
%$S[l..r]$.
If $S[l..r]$ overlaps 3 pieces or more, we choose its leftmost 
contained piece (the left and right endpoints of this piece are $\selop_1(P,\rankop_1(P,l-1)+1)$ and 
$\selop_1(P,\rankop_1(P,l-1)+2)-1$, respectively),
and do as before to obtain its $A$ to $3A$ candidates
and count their occurrences in $S$. This process yields, in time
$O(\frac{\log n}{\alpha \log\log n})$, an $\alpha$-minority of $S[l..r]$, if
there is one.

\subsection{Handling Updates}

To insert a new symbol $c$ at position $i$ in $S$, we first do the insertion in
$S$, and also insert a $0$ in $P[i]$ and $C[i]$. We then find the piece 
$P[x..y]$ where $P[i]$ belongs, using $x = \selop_1(P,\rankop_1(P,i))$ and 
$y = \selop_1(\rankop_1(P,i)+1)-1$. Finally, if $\rankop_c(S,y)-\rankop_c(S,x-1)
=1$, then $c$ is a new distinct symbol in the piece and we must mark it, with
$C[i] \leftarrow 1$.

This completes the insertion process unless we exceed the maximum number of 
distinct element in the piece, that is, $\rankop_1(C,y)-\rankop_1(C,x-1)>3A$.
In this case, we {\em repartition} the piece into pieces of size $A$ to $3A$.

The repartitioning proceeds as follows. We locate the first occurrences of
the distinct elements in the piece, using $\rankop$ and $\selop$ on $C$ and $S$.
We unmark (i.e., set to $0$) their piece positions in $C$ (which may not be
their first occurrences).
We then use the classic algorithm that computes order statistics in linear time
(i.e., $O(A)$) to find the $(A+1)$th of those first occurrence positions, $p$. 
Then the first piece, with $A$ distinct
elements, goes from $x$ to $p-1$, so we set $P[p] \leftarrow 1$ and mark in
$C$ the $A$ positions we had located. We then use $\rankop$ and
$\selop$ on $S$ to find the first occurrence of those $A$ positions in
$S[p..y]$, continue with the second piece, and so on.

This process generates a number of pieces with $A$ distinct elements, except
the last one, which may have fewer. In this case, we merge the last two pieces
built into one, which will have less than $2A$ distinct elements. Those are
found by repeating the generation of the penultimate piece, this time not
stopping at the $(A+1)$th smallest position, but rather including all the 
distinct elements.
Overall, each new piece is built in time $O(\frac{\log n}{\alpha \log\log n})$.

To delete $c=S[i]$, we remove the position $i$ from sequences $S$, $P$, and $C$.
If it holds that $P[i]=1$ before removing it, then we had deleted the mark of
the beginning of the piece, so we reset $P[i] \leftarrow 1$
again after removing $P[i]$.
If it holds that $C[i]=1$ before removing it, we must see if there is another
occurrence of $c$ in its piece. We compute the piece endpoints $x$ and $y$
as for the insertion, and then see if $\rankop_c(S,y)-\rankop_c(S,x-1) \ge 1$. 
If so, we set another occurrence of $c$ in $C$, for example,
$C[\selop_c(S,\rankop_c(S,x-1)+1)] \leftarrow 1$. Otherwise, we have lost a 
distinct element in the piece and must see if we still have sufficiently many
distinct elements, that is, if $\rankop_1(C,y)-\rankop_1(C,x-1) \ge A$. If this is true, we finish.

Otherwise, we have less than $A$ distinct elements in the piece, so we 
merge it with the previous or next piece (if none exists, then $S$ has only
one piece, which can have less than $A$ distinct elements). 
The merged piece has at least 
$A$ distinct elements, but it might have up to $4A-1$ and thus overflow. The 
merging then consists of removing the intermediate $1$ in $P$ that separates 
the two pieces and running our repartitioning process described above. The
cost will be, again, $O(\frac{\log n}{\alpha \log\log n})$ per piece generated.

\subsection{Analysis}

Our partitioning process may require time proportional to the length of the
piece, which can be arbitrarily longer than $3A$. Consider for example $A=2$
and the piece $(abc)^n def$. Inserting a $g$ at the end produces 
$\frac{3}{2}n+2$ pieces of length $2$. We can show, however, that the amortized
cost of a repartitioning is within the desired time bounds. We will measure 
the cost in terms of number of operations over the sequences, each of which costs 
$O(\frac{\log n}{\log\log n})$.

Let us define a potential function $\phi = n - A \cdot (m-1)$, where $m$ is the
number of pieces at the present moment. It always holds $\phi \ge 0$, even when
we start with $n=1$ element and $m=1$ piece. Then an insertion or a deletion 
without overflow or underflow modifies $\phi$ by $\Delta\phi = \pm 1$, which 
does not affect
the asymptotic cost. A repartitioning, instead, takes a piece and produces $t>1$
pieces out of it. The actual cost of generating each piece, ignoring constants,
is $A$ operations on the sequences;
therefore the total cost of the operation is $A \cdot t$. On the other hand,
the difference in potential is $\Delta\phi = A \cdot (1-t)$ ($n$ does not 
change while repartitioning). Therefore, the amortized cost of the 
repartitioning is $A$.

The case of an underflow is similar: we first join two pieces, which increases
$\phi$ by $A$. We then repartition the resulting piece into $t$, which costs
$A \cdot t$ and changes the potential by $\Delta\phi = A \cdot (1-t)$. In 
total, the amortized cost of an underflow is $2A$.

Since $A=O(1/\alpha)$ and the operations we are counting cost 
$O(\frac{\log n}{\log\log n})$, the amortized cost of the operations $\insop$ 
and $\delop$ is $O(\frac{\log n}{\alpha \log\log n})$.

\begin{theorem} \label{thm:minority}
  For any $0<\alpha<1$, a sequence of length $n$ over an alphabet of size $\sigma$ can be represented using $nH_k+ 2n + o(n\lg \sigma)$ bits for any $k = o(\log_{\sigma} n)$ to answer range $\alpha$-minority queries in $O(\frac{\lg n}{\alpha \lg\lg n})$ time, and to support symbol insertions and deletions in $O(\frac{\lg n}{\alpha \log\log n})$ amortized time.
\end{theorem}

By using the idea in static form, we also obtain the first solution for
$\alpha$-minority queries using high-order entropy space. Here the bitvectors
take constant time to answer $\accop$, $\rankop$ and $\selop$ queries, and the
static sequence representation \cite[Thm. 11]{BN15} yields time
$O(\lg\lg_w\sigma)$ per operation. The construction is easily done in linear
time.

\begin{theorem}
On a RAM machine of $w = \Omega(\lg n)$ bits, a sequence of length n over an 
alphabet of size $\sigma$ can be represented using $nH_k + 2n + o(n\lg\sigma)$ 
bits for any $k = o(\lg_\sigma n)$, to answer range $\alpha$-minority queries 
in $O((1/\alpha) \lg\lg_w\sigma)$ time. The structure is built in $O(n)$ time.
\end{theorem}

\section{Conclusions}
\label{sec:conclusions}

In this article, we have designed the first compressed data structure for dynamic range $\alpha$-majority.
To achieve this result, our key strategy is to perform a multi-level
decomposition of the sequence $S$ and, for each block of $S$, precompute a
candidate set that includes all the $\alpha$-majorities of any query range of the right size that touches this block.
Thus, when answering a query, we need not find a set of blocks whose union
forms the query range, as is required in the solution of Elmasry~\etal~\cite{ehmn2016}. Instead, we only look for a single block that touches the query range.
  This simpler strategy allows us to achieve compressed space.
  
  Furthermore, we generalize our solution to design the first dynamic data
structure that can  maintain $S$ in the same space and update time, to support
the computation of the $\beta$-majorities in a given query range for any
$\beta \in [\alpha, 1)$ in $O(\frac{\lg n}{\beta\lg\lg n})$ time. Note that
here $\beta$ is given with the queries, and only $\alpha$ is fixed and given beforehand. This type of query is more general than range $\alpha$-majority queries and was only studied in the static case before~\cite{ghmn2011,bgmnn2016}.
  
Finally, we design the first dynamic data structure for the range $\alpha$-minority query problem, and this data structure is also compressed. Even simple
static solutions \cite{cdsw2015} based on range minimum queries are difficult
to dynamize. We find a new, simple data structure that is easy to maintain upon
updates and gives sufficient information to find $\alpha$-minorities in time
$O(\frac{\lg n}{\alpha\lg\lg n})$.

For constant $\alpha$, our query time $O(\frac{\lg n}{\lg\lg n})$ for $\alpha$-majority is optimal within polylogarithmic update time, as it matches the lower bound of the simpler operation {\em majority} \cite[Prop.\ 3]{HR03}, which considers the particular case of binary alphabets, ranges of the form $S[1..i]$, and $\alpha=1/2$. Another obvious lower bound is $O(1/\alpha)$, as it is the output size in the worst case. It is not clear whether a dynamic structure can achieve $O(\frac{\lg n}{\lg\lg n} + \frac{1}{\alpha})$ query time with polylogarithmic update time, even without compression. In the case of $\alpha$-minorities, there is no clear lower bound.

Interestingly, our dynamic data structures for $\alpha$-majority and 
$\alpha$-minority use $nH_k+o(n\log\sigma)$ bits of space ($+2n$ bits in the
case of minorities), which is less than the space achieved so far in the static
case. We thus describe static variants of our structures, which answer queries
in time $O((1/\alpha)\log\log_w \sigma)$. The best static solutions answer 
queries in time $O(1/\alpha)$ (which is optimal at least for
$\alpha$-majorities), but use more than $nH_0 \ge nH_k$
bits of space \cite{bgmnn2016}. An interesting question is whether the optimal 
times can be achieved within $k$-th order entropy space. 

\bibliographystyle{splncs03}
\bibliography{compressedmajority}

%\Newpagex
%\appendix
%\begin{center}
%  \bf \Large Appendices
%\end{center}

\end{document}